\documentclass[a4paper,11pt]{article}

\usepackage[utf8]{inputenc}
\usepackage[backend=biber,style=alphabetic,maxbibnames=100,maxalphanames=4]{biblatex}
\usepackage[in]{fullpage}
\usepackage[dvipsnames,svgnames,x11names,hyperref]{xcolor}
\usepackage[colorlinks]{hyperref}
\usepackage{enumerate}
\usepackage{nicefrac}
\usepackage{graphicx}
\usepackage{algorithm}
\usepackage{amsmath}
\usepackage{mathtools}
\usepackage{authblk}
\usepackage{amssymb}
\usepackage{amsthm}
\usepackage{esvect}
\usepackage[normalem]{ulem}
\usepackage{braket}
\usepackage{pgfplots}
\usepackage{multicol}
\usepackage[final]{microtype}
\usepackage[font=small]{caption}

\makeatletter
\newtheorem*{rep@theorem}{\rep@title}
\newcommand{\newreptheorem}[2]{%
\newenvironment{rep#1}[1]{%
 \def\rep@title{#2 \ref{##1}}%
 \begin{rep@theorem}}%
 {\end{rep@theorem}}}
\makeatother
\newreptheorem{theorem}{Theorem}

\pgfplotsset{compat=1.6}

\bibliography{bibliography}

\hypersetup{
	colorlinks = true,
	citecolor = teal,
	urlcolor = purple,
}

\newtheorem{theorem}{Theorem}
\newtheorem{lemma}[theorem]{Lemma}

\newtheorem{corollary}[theorem]{Corollary}

\DeclareMathOperator{\polylog}{polylog}

\AtEveryBibitem{
    \clearlist{location}
    \clearfield{venue}
    \clearname{editor}
}

\title{Quantum algorithms for Hopcroft's problem}
\author{Vladimirs Andrejevs}
\author{Aleksandrs Belovs}
\author{Jevgēnijs Vihrovs}

\affil{Centre for Quantum Computer Science, Faculty of Computing,\authorcr University of Latvia, Raiņa 19, Riga, Latvia, LV-1586}

\date{}

\begin{document}

\maketitle

\begin{abstract}
In this work we study quantum algorithms for Hopcroft's problem which is a fundamental problem in computational geometry.
Given $n$ points and $n$ lines in the plane, the task 
is to determine whether there is a point-line incidence.
The classical complexity of this problem is well-studied, with the best known algorithm running in $O(n^{4/3})$ time, with matching lower bounds in some restricted settings.
Our results are two different quantum algorithms with time complexity $\widetilde O(n^{5/6})$.
The first algorithm is based on partition trees and the quantum backtracking algorithm.
The second algorithm uses a quantum walk together with a history-independent dynamic data structure for storing line arrangement which supports efficient point location queries.
In the setting where the number of points and lines differ, the quantum walk-based algorithm is asymptotically faster.
The quantum speedups for the aforementioned data structures may be useful for other geometric problems.
\end{abstract}

\section{Introduction} \label{sec:intro}

In this work we investigate the quantum complexity of Hopcroft's problem, a classic problem in computational geometry.
Given $n$ lines and $n$ points in the plane, it asks to determine whether some point lies on some line.
In a line of research spanning roughly 40 years culminating with a recent paper by Chan and Zheng \cite{CZ23a}, the classical complexity has settled on $O(n^{4/3})$ time, with matching lower bounds in some models of computation \cite{Eri96}.
Along with its natural setting, the problem also captures the essence of a class of other geometric problems with similar complexity \cite{Eri95}.

There are several reasons why we find Hopcroft's problem interesting in the quantum setting.
Firstly, classical algorithms for this problem typically use data structures supporting some fundamental geometric query operations.
For example, Hopcroft's problem can be reduced to the \emph{simplex range searching}, in which the data structure stores the given points and each query asks whether a given region contains any of the points \cite{Aga17}.
Another approach is to store the given lines instead and answer \emph{point location queries}, that is, for a given point, determine which region of the line configuration it belongs to \cite{ST86,Ede87}.
Thus, Hopcroft's problem gives a good playground for improving and comparing the complexity of ubiquitous geometric data structures quantumly.
We are also interested in finding new \emph{history-independent} data structures that can be used in quantum walk algorithms, following the ideas started with Ambainis' element distinctness algorithm \cite{Amb03} (see also \cite{ACHWZ20,BLPS22b}).

Secondly, Hopcroft's problem is closely related to a large group of geometric tasks.
Many problems can be solved in the same time $\widetilde O(n^{4/3})$, and a speedup for Hopcroft's problem may automatically give an improvement for those.
Erickson studied the class of such problems \cite{Eri95}, some examples include detecting/counting intersections in a set of segments and detecting/counting points in a given set of regions.
The problem can be also reduced to various other geometric problems, giving fine-grained lower bounds.
For example, Hopcroft's problem in $d$ dimensions (replacing lines with hyperplanes) can be reduced to halfspace range checking in $d+1$ dimensions for $d \geq 4$ (are all given points above all given hyperplanes?) and others \cite{Eri95}.

In fact, Hopcroft's problem in $d$ dimensions can also be equivalently formulated as follows: given two sets of vectors $\mathcal A, \mathcal B \in \mathbb R^{d+1}$, determine whether there are $a \in \mathcal A$, $b \in \mathcal B$ such that $\langle a, b \rangle = 0$ \cite{WY14}.
The famous \textsc{Orthogonal Vectors} problem (OV) in fine-grained complexity is a special case of Hopcroft's when $\mathcal A, \mathcal B = \{0,1\}^d$ \cite{Wil05,AWW14}.
The complexities of these problems differ; if $|A|=|B|=n$, then classically, the complexity of OV is $\Theta(n)$ in $O(1)$ dimensions \cite{Wil17} and $\Theta(n^{2-o(1)})$ in $\polylog n$ dimensions under SETH \cite{AWY15,CW21}.
In contrast, the complexity of Hopcroft's problem in $d$ dimensions is $O(n^{2d/(d+1)})$ \cite{CZ23a}.
Quantumly, the complexity of OV was settled in \cite{ACHWZ20}; for $O(1)$ dimensions, it is $\Theta(\sqrt{n})$ and for $\polylog n$ dimensions, $\Theta(n^{1+o(1)})$ under QSETH, the quantum analogue of SETH.
In this work we also examine the quantum complexity of Hopcroft's problem in an arbitrary number of dimensions $d$.

In general, we are interested in investigating quantum speedups for computational geometry problems.
In recent years, there have been several works researching this topic.
First, Ambainis and Larka gave a nearly optimal $O(n^{1+o(1)})$ quantum algorithm for the \textsc{Point-On-3-Lines} (detecting whether three lines are concurrent among the $n$ given) and similar problems \cite{AL20}.
This problem is closely connected to fine-grained complexity as well, as it is an instance of the \textsc{3-Sum-Hard} problem class.
Classically, it is conjectured that \textsc{3-Sum} cannot be solved faster than $O(n^2)$; the authors also conjectured a quantum analogue that \textsc{3-Sum} cannot be solved quantumly faster than $O(n)$, and Buhrman et al.~used this conjecture to prove conditional quantum lower bounds on various geometrical problems \cite{BLPS22a}.
Aaronson et al.~studied the quantum complexity of the \textsc{Closest Pair} problem (finding the closest pair of points among the $n$ given), proving an optimal $\widetilde \Theta(n^{2/3})$ running time in $O(1)$ dimensions using a quantum walk algorithm with a dynamic history-independent data structure for storing points that is able to detect $\epsilon$-close pairs of points \cite{ACHWZ20}.
For the \textsc{Bipartite Closest Pair} problem in $d$ dimensions (finding the closest pair of points between two sets of size $n$), they gave an $O(n^{1-1/{2d}+\delta})$ time quantum algorithm for any $\delta > 0$.
For other results in quantum algorithms for computational geometry, see \cite{SST01,SST02,BDLK06,VM09,VM10,KMM24}.

For Hopcroft's problem, the best classical results give complexity $O((nm)^{2/3} + m \log n + n \log m)$, where $n$ and $m$ are the number of lines and points, respectively.
In $d$ dimensions, these generalize to $O((nm)^{d/(d+1)} + m \log n + n \log m)$ complexity \cite{CZ23a}.
The first complexity is unconditionally optimal if the algorithm needs to list all incidences, since there exists a planar construction with $\Omega((nm)^{2/3})$ incidence pairs (\cite{Ede87}, Section 6.5.).
It is also believed to be optimal for detection as well, with matching lower bounds in some models \cite{Eri96}.
The dependence on $n$ and $m$ is symmetric since Hopcroft's problem is self-dual, in the sense that there is a geometric transformation which maps lines to points and vice versa, while preserving the point-line incidences (\cite{Ede87}, Section 14.3.).

Finally, quite often the quantum query complexity of a problem matches its time complexity, like in \textsc{Unstructured Search} \cite{Grover96,BBBV97}, \textsc{Element Distinctness} \cite{Amb07,Amb05,Kut05}, \textsc{Closest Pair} \cite{ACHWZ20}, \textsc{Claw Finding} \cite{Tan09,Zha05}, just to name a few.
In other cases even the precise query complexity is not yet clear, for example, \textsc{Triangle Finding} \cite{LeG14} or \textsc{Boolean Matrix Product Verification} \cite{BŠ06,CKK12}.
In the case of Hopcroft's problem, its quantum query complexity can be easily characterized to be $\Theta((nm)^{1/3} + \sqrt{n} + \sqrt{m})$ from known results.
The query-efficient algorithm does not immediately generalize to time complexity; therefore, the main focus here is on improving the performance of the relevant classical data structures quantumly.

\subsection{Our results}

In this work we show two quantum algorithms for Hopcroft's problem with time complexity $\widetilde O(n^{5/6})$.
This constitutes a polynomial speedup over the classical $O(n^{4/3})$ time.
We obtain our results by speeding up classical geometric data structures using different quantum techniques.

Note that it is actually not hard to achieve a speedup over the classical complexity $O(n^{4/3})$, as we can simply use Grover's algorithm to search over the line and point pairs, obtaining a $\widetilde O(n)$ algorithm.
The key to getting faster speedups is to utilize the known data structures for geometric queries.
For example, it is also not hard to get an even better quantum algorithm: split all lines into $n/k$ groups of size $k$ and use Grover search over these groups.
For each group, one can build a classical data structure for point location in a line arrangement with preprocessing time $\widetilde O(k^2)$ and $\polylog k$ query time (see e.g.~\cite{Ede87}).
Then one can use another Grover's search over the $n$ points to test whether any of them is located on some of the $k$ lines of the group.
By choosing the optimal $k$, one obtains the complexity $\widetilde O(\sqrt{n/k}(k^2 + \sqrt{n})) = \widetilde O(n^{7/8})$.

To obtain better complexity, we wish to speed up classical geometric data structures directly.
We look at two underlying fundamental problems one usually encounters on the way to solving Hopcroft's problem.
\begin{enumerate}
    \item \textbf{Simplex range searching.} In simplex range searching, the input is a set of $n$ points in $d$-dimensional space.
    A query then asks whether a given simplex contains any of the given points.
    The query may also ask to list or count the points in the simplex, among other variants \cite{Aga17}.
    Usually, there is some \emph{preprocessing time} to precompute the data structure and some \emph{query time} to answer each query.
    Classically, these complexities are well-understood; in a nutshell, a data structure of size $m$ can be constructed in $\widetilde O(m)$ time and each query can then be answered in $\widetilde O(n/m^{1/d})$ time \cite{Mat93}, and this is matched by lower bounds in the semigroup model \cite{Cha89,CB96}.
    If the allowed memory size is linear, then preprocessing and query times become respectively $O(n \log n)$ and $O(n^{1-1/d})$ \cite{Cha12}.
    In this paper we require a variant of simplex range queries which we call \emph{hyperplane emptiness queries}, where we have to determine whether a query hyperplane contains any of the given points.
    We show that quantumly we can speed up query time quadratically by using Montanaro's quantum algorithm for searching in backtracking trees \cite{Montanaro15}:
    \begin{reptheorem}{thm:quantum-query}
    There is a bounded-error quantum algorithm that can answer hyperplane emptiness queries in $O(\sqrt{n^{1-1/d}\cdot \log n})$ time, with $O(n \log n)$ preprocessing time.
    \end{reptheorem}
    We note that this result is not really specific to the hyperplane emptiness queries, as all we are doing is speeding up search in the \emph{partition tree} data structure \cite{Cha12}, thus this result can applied to other types of queries as well.
    However, this speedup does not extend to the counting version of simplex queries, since, essentially, our procedure implements a search for a marked vertex in a tree using quantum walk.
    With this result, we show a quantum speedup for Hopcroft's problem in $d$ dimensions:
    \begin{reptheorem}{thm:speedup-backtracking}
    There is a bounded-error quantum algorithm which solves Hopcroft's problem with $n$ hyperplanes and $m \leq n$ points in $d$ dimensions in time:
    \begin{itemize}
        \item $\widetilde O(n^{\frac{d}{2(d+1)}}m^{1/2})$, if $m \geq n^{\frac{d}{d+1}}$;
        \item $\widetilde O(n^{1/2}m^{\frac{d-1}{2d}})$, if $m \leq n^{\frac{d}{d+1}}$.
    \end{itemize}
    If $n = m$, then the algorithm has complexity $\widetilde O(n^{1-\frac{1}{2(d+1)}})$.
    \end{reptheorem}
    In particular, the complexity is $\widetilde O(n^{5/6})$ in $2$ dimensions for $n = m$.
    \item \textbf{Planar point location.}
    The second approach is to use point location queries.
    As mentioned in the example algorithm above, one can use classical point location data structures to determine whether a query point lies on the boundary of the region it belongs to.
    More specifically, we consider only data structures for \emph{planar} point location in \emph{line arrangements}.
    A set of $n$ lines partitions the plane into $O(n^2)$ regions; this an old and well-researched topic, with many approaches to construct a data structure that holds the description of these regions in $O(n^2)$ time, the same amount of space and $\polylog n$ point location query time \cite{Ede87} (in fact, $O(n^d)$ preprocessing time and space and $O(\log n)$ query time in $d$ dimensions \cite{Cha93,CZ23b}).
    
    In addition, there are also dynamic data structures with the same preprocessing and query times.
    More specifically, one can insert or remove a line in time $O(n)$ (or $O(n^{d-1})$ in $d$ dimensions) \cite{MS92}.
    We take an opportunity to employ such a data structure in a \emph{quantum walk} algorithm on a Johnson graph to solve Hopcroft's problem.
    In particular, we develop a history-independent randomized data structure for storing an arrangement of an $r$-subset of $n$ lines with the ability to perform line insertion/removal in $O(r \polylog n)$ time and point location in $\polylog n$ time, requiring $O(r \polylog n)$ memory storage.

    To do that, we store $k$-levels of the line arrangement in history-independent skip lists a la Ambainis \cite{Amb07}.
    A $k$-level of a line arrangement is a set of segments of lines such that there are exactly $k$ lines above each edge.
    It turns out that skip lists are ideal for encoding the $k$-levels.
    For example, when a new line is inserted, it splits each $k$-level in two parts, one of which will still belong to the $k$-level, but the other will belong to the $(k+1)$-level.
    We can then ``reglue'' these two tails to the correct levels of the arrangements in $\polylog n$ time by utilizing the properties of the skip list, all while keeping the history independence of the data structure.
    Figure \ref{fig:reglue} illustrates this idea with two polygonal chains.
    \begin{minipage}{\linewidth}
    \begin{figure}[H]
        \centering
        \includegraphics[scale=1.3]{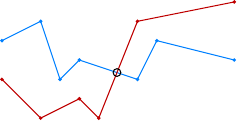} \hspace{1cm} \includegraphics[scale=1.3]{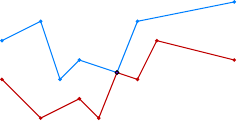}
        \caption{Examine two polygonal chains of length at most $n$ that are stored in skip lists, where the $x$ coordinates of the points are in increasing order.
        Suppose that the chains intersect at a single point, which we know.
        Then we can swap the tails of the skip lists after the intersection point efficiently.
        The procedure first finds the position of the intersection point in the skip lists in $\polylog n$ time.
        Then it adds the intersection point to both of the skip lists and swaps the pointers directed to points after the intersection, which also takes $\polylog n$ time.}
        \vspace{10pt}
        \label{fig:reglue}
    \end{figure}
    \end{minipage}

    Using this data structure, we show the following quantum speedup for Hopcroft's problem in $2$ dimensions:
    \begin{reptheorem}{thm:algo2}
    There is a bounded-error quantum algorithm that solves Hopcroft's problem with $n$ lines and $m \leq n$ points in the plane in time:
    \begin{itemize}
        \item $\widetilde O(n^{1/3} m^{1/2})$, if $n^{2/3} \leq m$;
        \item $\widetilde O(n^{2/5} m^{2/5})$, if $n^{1/4} \leq m \leq n^{2/3}$;
        \item $\widetilde O(n^{1/2})$, if $m \leq n^{1/4}$.
    \end{itemize}
    In particular, the complexity is $\widetilde O(n^{5/6})$ when $n = m$.
    \end{reptheorem}
\end{enumerate}

Both of Theorems \ref{thm:speedup-backtracking} and \ref{thm:algo2} have their pros and cons.
Theorem \ref{thm:speedup-backtracking} is arguably simpler, since it is a quite direct application of the quantum speedup for backtracking.
It also has a lower polylogarithmic factor hidden in the $\widetilde O$ notation, only $\log n$ compared to $\log^6 n$ in Theorem \ref{thm:algo2}.
However, Theorem \ref{thm:algo2} gives better asymptotic complexity if the number of lines $n$ differ from the number of points $m$.
On the other hand, Theorem \ref{thm:speedup-backtracking} gives a speedup in the case of an arbitrary number of dimensions, while Theorem \ref{thm:algo2} has something to say only about the planar case; we leave a possible generalization of this approach to larger dimensions for future research.

\section{Preliminaries}

We assume that the sets of points and lines are both in a general position (no two lines are parallel, no three lines intersect at the same point, no three points lie on the same line).

\subsection{Model}

We use the standard quantum circuit model together with Quantum Random Access Gates (QRAG) (see, for example, \cite{ABDLS23}).
This gate implements the following mapping:
\[\ket{i}\ket{b}\ket{x_1,\ldots,x_N} \mapsto \ket{i}\ket{x_i}\ket{x_1,\ldots,x_{i-1},b,x_{i+1},\ldots,x_N}.\]
Here, the last register represents the memory space of $N$ bits.
Essentially, QRAG gates allow both for reading memory in superposition as well as writing operations.
We note that both our algorithms require ``read-write'' quantum memory, so it is not sufficient to use the weaker ``read-only'' QRAM gate, which is enough for some quantum algorithms.

To keep the analysis of the algorithms clean, we abstract the complexity of basic underlying operations under the ``unit cost''.
\begin{itemize}
    \item basic arithmetic operations on $O(\log n)$ bits;
    \item the implementation of QRAG and elementary gates;
    \item the running time of a quantum oracle, which with a single query can return the description of any point or line.
\end{itemize}
In the end, we measure the time complexity in the total amount of unit cost operations.
The unit cost can be taken as the largest running time of the operations listed above, which will add a multiplicative factor in the complexity.

Assuming that an application of a QRAG gate takes unit time is also useful for utilizing the existing classical algorithms in the RAM model.
The classical algorithms that we use work in the real RAM model, where arithmetic operations and memory calls on $O(\log n)$ bits are considered to be executed in constant time.
Thus, we work in the quantum analogue of real RAM, and if there is a time $T$ classical real RAM algorithm, then we can use it in time $O(T)$ in this model.
The actual implementation of QRAG is an area of open research and debate; however, there exist theoretical proposals that realize such operations in time polylogarithmic in the size of the memory, like the bucket brigade architecture of \cite{GLM08}.

\subsection{Tools}

One of the building blocks in our algorithms is the following version of Grover's search:

\begin{theorem}[Grover's search with bounded-error inputs \cite{ABBLS23, HMDw03}] \label{thm:grover-with-errors}
    Let $\mathcal A : [N] \to \{0,1\}$ be a bounded-error quantum procedure with running time $T$.
    Then there exists a bounded-error quantum algorithm that computes $\bigvee_{i \in [N]} \mathcal A(i)$ with running time $O(\sqrt{N}( T + \log N))$.
\end{theorem}

Effectively, this result states that even if the inputs to Grover's search are faulty with constant probability, no error boosting is necessary, which would add another logarithmic factor to the complexity.
We say that an algorithm is \emph{bounded-error} if its probability of incorrect output is some constant strictly less than $1/2$.

\section{Query complexity}

Before examining time-efficient quantum algorithms, we take a look at the quantum query complexity of Hopcroft's problem, which in this case can be fully characterized.

\begin{theorem} \label{thm:lower-bound}
    The quantum query complexity of Hopcroft's problem on $n$ lines and $m$ points in two dimensions is $\Theta(n^{1/3}m^{1/3}+\sqrt{n}+\sqrt{m})$.
\end{theorem}

\begin{proof}
    For the upper bound, Hopcroft's problem can be seen as an instance of the bipartite \emph{subset-finding} problem, in which one is given query access to two sets $X$ and $Y$ of sizes $n$ and $m$, respectively, and needs to detect whether there is a pair $x \in X$, $y \in Y$ satisfying some predicate $R : X \times Y \to \{0,1\}$.
    For this problem, Tani gave a quantum algorithm with query complexity $O(n^{1/3}m^{1/3}+\sqrt{n}+\sqrt{m})$ \cite{Tan09}.

    For the lower bound, we can reduce the bipartite element distinctness problem (also known as \textsc{Claw Finding}) to Hopcroft's problem.
    In this problem, we have two sets of variables $x_1, \ldots, x_n \in [N]$ and $y_1, \ldots, y_m \in [N]$ and we need to detect whether $x_i = y_j$ for some $i$, $j$.
    Zhang proved that for this problem $\Omega(n^{1/3}m^{1/3}+\sqrt{n}+\sqrt{m})$ quantum queries are needed \cite{Zha05}.
    We reduce each $x_i$ to a line $x = x_i$ and each $y_j$ to a point $(y_j,0)$.
    Then $x_i = y_j$ only iff some point belongs to some line, so the statement follows.
\end{proof}

In particular, this proves that Theorem \ref{thm:algo2} is asymptotically optimal for $m \leq n^{1/4}$.
The query complexity and the complexities of our algorithms are shown graphically in Figure \ref{fig:complexity}.

\begin{figure}[H]
    \centering\begin{tikzpicture}[scale=0.8]
        \begin{axis}[
            mark = *,
            axis lines = left,
            xlabel={$\log_n m$},
            ylabel={$\log_n(\text{time})$},
            xmin=0, xmax=1,
            ymin=0.4, ymax=1,
            xtick={0,0.25,0.5,0.75,1},
            ytick={0.4,0.6,0.8,1},
            %xmajorgrids=true,
            grid style=dashed,
            legend pos=north west
        ]
        
        \addplot[color=RoyalBlue, line width = 1pt] coordinates {
            (0,1/2)
            (2/3,2/3+0.004)
            (1,5/6+0.004)
        };
        
        \addplot[color=ForestGreen, line width = 1pt] coordinates {
            (0,1/2)
            (1/4,1/2)
            (2/3,2/3)
            (1,5/6)
        };
        
        \addplot[color=OrangeRed, line width = 1pt] coordinates {
            (0,1/2-0.004)
            (1/4,1/2-0.004)
            (1,2/3)
        };
        \end{axis}
        \end{tikzpicture}
    \caption{The quantum time complexity of Hopcroft's problem in $2$ dimensions on $m$ points and $n$ lines, assuming $m \leq n$. The red line shows the query complexity (Theorem \ref{thm:lower-bound}); the blue line shows the complexity of the quantum algorithm based on the partition tree (Theorem \ref{thm:speedup-backtracking}); the green line shows the complexity of the quantum walk algorithm with the line arrangement data structure (Theorem \ref{thm:algo2}).}
    \label{fig:complexity}
\end{figure}
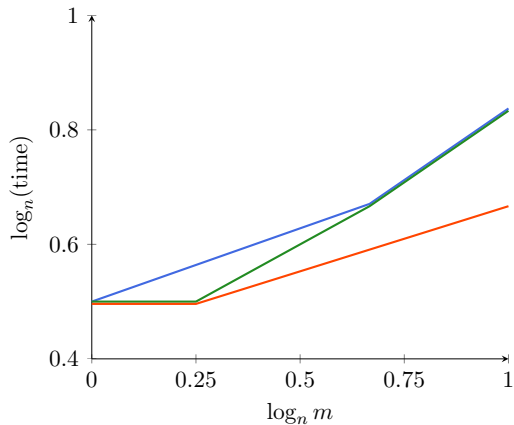

There is an obvious hurdle in implementing the query algorithm of Tani in the same time complexity here.
Their quantum walk would require a history-independent dynamic data structure for storing a set of lines and a set of points supporting the detection of incidence existence.
Even assuming the most optimistic quantum versions of known data structures, a sufficiently powerful speedup looks unfeasible.
The next two sections describe the algorithms we have obtained by speeding up two fundamental geometrical query data structures quantumly.

\section{Algorithm 1: quantum backtracking with partition trees}

We begin with a brief overview of the classical partition tree data structure and its quantum speedup, and then proceed with the description of the quantum algorithm. 

\subsection{Partition trees}

The partition tree is a classical data structure designed for the task of \emph{simplex range searching}.
In this problem, one is given $n$ points in a $d$-dimensional space; the task is to answer queries where the input is a simplex and the answer is the number of points inside that simplex.
The \emph{partition tree} is a well-known data structure which can be used to solve this task \cite{Aga17}.

This data structure can be described as a tree in the following way.
The tree stores $n$ points and each subtree stores a subset of these points.
Each interior vertex $v$ stores a simplex $\Delta(v)$ such that all of the points stored in this subtree belong to the interior of $\Delta(v)$.
For each interior vertex $v$, the subsets of the points stored in its children form a partition of the points stored in the subtree of $v$.
Each leaf vertex stores a constant number of points in a list.
For our purposes, each interior vertex stores only with information about its children and no information about the points that are stored in its subtree.

In this work, we are interested in the \emph{hyperplane emptiness queries}.
Given $n$ points in the $d$-dimensional space, the task is to answer queries where the input is an arbitrary hyperplane and the answer is whether there is a point that lies on the given hyperplane.
These queries can also be answered using partition trees:

\begin{lemma}[Hyperplane emptiness query procedure] \label{thm:line-query}
Let $\mathcal T$ be a partition tree such that each vertex has only a constant number of children.
Let the \emph{tree query cost} $c(\mathcal T)$ be the maximum number of simplices of $\mathcal T$ that intersect an arbitrary hyperplane.
Then the hyperplane emptiness query can be answered in $O(c(\mathcal T))$ time.
\end{lemma}

\begin{proof}
    The procedure for answering a query is as follows.
    We start at the root and traverse $\mathcal T$ recursively.
    If the current vertex is an interior vertex $v$ and the query hyperplane is $h$, then we recurse only in the children of $v$ such that $\Delta(v)$ intersects $h$.
    If the current vertex is a leaf vertex, we check whether any of its points lies on $h$.
    The running time is evidently linear in the number of simplices intersecting $h$.
\end{proof}

There are different ways to construct partition trees, but a long chain of works in computational geometry resulted in an optimal version of the partition tree \cite{Cha12}.
Even though their goal is to answer simplex queries, in fact the main result gives an upper bound on $c(\mathcal T)$ for their partition tree:
\begin{theorem}[Partition tree \cite{Cha12}] \label{thm:partition-tree}
    For any set of $n$ points in $d$ dimensions, there is a partition tree $\mathcal T$ such that:
    \begin{itemize}
        \item it can be built in $O(n \log n)$ time and requires $O(n)$ space;
        \item $c(\mathcal T) = O(n^{1-1/d})$; hence, a hyperplane emptiness query requires $O(n^{1-1/d})$ time;
        \item each vertex has $O(1)$ children and the depth of the tree is $O(\log n)$.
    \end{itemize}
\end{theorem}

To speed up the emptiness query time of the partition tree quantumly we use the quantum backtracking algorithm \cite{Montanaro15}.
Their quantum algorithm searches for a \emph{marked} vertex in a tree $\mathcal S$.
The markedness is defined by a black-box function $P : V(\mathcal T) \to \{\textsc{true}, \textsc{false}, \textsc{indeterminate}\}$.
For leaf vertices $v$, we have $P(v) \in \{\textsc{true}, \textsc{false}\}$.
A vertex $v$ is marked if $P(v) = \textsc{true}$, and the task is to determine whether $\mathcal S$ contains a marked vertex.

The root of $\mathcal S$ is known and the rest of the tree is given by two other black-box functions.
The first, given a vertex $v$, returns the number of children $d(v)$ of $v$.
The second, given $v$ and an index $i \in [d(v)]$, returns the $i$-th child of $v$.
The main result is a quantum algorithm for detecting a marked vertex in $\mathcal S$:
\begin{theorem}[Quantum algorithm for backtracking \cite{Montanaro15,AK17}] \label{thm:backtracking}
    Suppose you are given a tree $\mathcal S$ by the black boxes described above and upper bounds $T$ and $h$ on the size and the height of the tree.
    Additionally suppose that each vertex of $\mathcal S$ has $O(1)$ children.
    Then there is a bounded-error quantum algorithm that detects a marked vertex in $\mathcal S$ with query and time complexity $O(\sqrt{Th})$.
\end{theorem}

When we apply it to the partition tree from Theorem \ref{thm:partition-tree}, we get:
\begin{theorem} \label{thm:quantum-query}
    There is a bounded-error quantum algorithm that can answer hyperplane emptiness queries in $O(\sqrt{n^{1-1/d}\cdot \log n})$ time, with $O(n \log n)$ preprocessing time.
\end{theorem}

\begin{proof}
The procedure of Lemma \ref{thm:line-query} examines a tree $\mathcal S$, which is a subgraph of $\mathcal T$.

We will apply Theorem \ref{thm:backtracking} to $\mathcal S$.
Suppose that $h$ is a query hyperplane.
The black box $P$ returns \textsc{intermediate} for any interior vertex $v$ and for a leaf vertex $v$ returns \textsc{true} iff some point stored in $v$ lies on $h$.
The second black box returns the number of children of $v$ in $\mathcal T$ if $\Delta(v)$ intersects $h$ and $0$ otherwise.
The black box returning the $i$-th child simply fetches it from the partition tree that is stored in memory.
All of these black boxes require only constant time to implement.
Since we know that $|\mathcal S| = O(n^{1-1/d})$ and the height of $\mathcal S$ is $O(\log n)$ from Theorem \ref{thm:partition-tree}, then there is a quantum algorithm that solves the problem in $O(\sqrt{n^{1-1/d}\cdot \log n})$ time by Theorem \ref{thm:backtracking}.
\end{proof}

\subsection{Quantum algorithm}

Now we can apply the previous theorem to Hopcroft's problem.

\begin{theorem} \label{thm:speedup-backtracking}
    There is a bounded-error quantum algorithm which solves Hopcroft's problem with $n$ hyperplanes and $m \leq n$ points in $d$ dimensions in time:
    \begin{itemize}
        \item $O(n^{\frac{d}{2(d+1)}}m^{1/2}\log n)$, if $n^{\frac{d}{d+1}} \leq m$;
        \item $O(n^{1/2}m^{\frac{d-1}{2d}}\log n)$, if $m \leq n^{\frac{d}{d+1}}$.
    \end{itemize}
    If $n = m$, then the algorithm has complexity $O(n^{1-\frac{1}{2(d+1)}}\sqrt{\log n})$.
\end{theorem}

\begin{proof}
    In the first case, we partition the whole set of points into $m/r$ groups of size $r = n^{\frac{d}{d+1}}$.
    Using Grover's search, we search for a group that contains a point belonging to some line.
    To determine whether it's true for a fixed group, first we build a partition tree from Theorem \ref{thm:partition-tree} to store these points.
    Then we run Grover's search over all lines and determine whether a line contains some point from the group using the quantum query procedure from Theorem \ref{thm:quantum-query}.
    Overall, the complexity of this algorithm
    \begin{align*}
        O\left(\sqrt{\frac{m}{r}}\left(r\log r + \sqrt n \cdot \sqrt{r^{1-1/d}\cdot \log r}\right)\right)
        &= O\left(\left(\sqrt{mr}+\sqrt{\frac{nm}{r^{1/d}}}\right)\log r\right)\\
        &= O\left( \sqrt{m n^{\frac{d}{d+1}}}\log n\right).
    \end{align*}
    Here, we use Theorem \ref{thm:grover-with-errors} for Grover's search, and additional logarithmic terms get subsumed asymptotically, so we omit them for simplicity.
    
    If $m \leq n^{\frac{d}{d+1}}$, then we simply build the partition tree of all $m$ points, then use Grover's search over all lines and query the partition tree for each of them.
    The complexity in that case is
    \[O\left(m\log m + \sqrt{n}\cdot \left(\sqrt{m^{1-1/d}\cdot \log m} + \log n\right)\right) = O\left(\sqrt{n m^{1-1/d}}\cdot \log n\right),\]
    because the second term dominates the first (up to logarithmic factors).
    The $\log n$ term comes from Theorem \ref{thm:grover-with-errors}.
\end{proof}

In two dimensions we have the following complexities:

\begin{corollary} \label{col:algo1}
    There is a bounded-error quantum algorithm which solves Hopcroft's problem with $n$ lines and $m \leq n$ points in the plane in time:
    \begin{itemize}
        \item $O(n^{1/3}m^{1/2}\log n)$, if $n^{2/3} \leq m$;
        \item $O(n^{1/2}m^{1/4}\log n)$, if $m \leq n^{2/3}$.
    \end{itemize}
    If $n = m$, the complexity is $O(n^{5/6}\log n)$.
\end{corollary}

\section{Algorithm 2: quantum walk with line arrangements}

First we describe a classical history-independent data structure for storing an arrangement of a set of lines.
After that, we describe the quantum walk algorithm that uses it for solving Hopcroft's problem.

\subsection{Line arrangements}

We begin with a few definitions (for a thorough treatment, see e.g.~\cite{Ede87}).
For a set of lines $L$, the \emph{line arrangement} $\mathcal A(L)$ is the partition of the plane into connected regions bounded by the lines.
The convex regions with no other lines crossing them are called \emph{cells} and their sides are called the \emph{edges} of the arrangement (note that some cells may be infinite).
The intersection points of the lines are called the \emph{vertices} of the arrangement.

For a set of lines $L$ in a general position, the \emph{$k$-level} is the set of edges of $\mathcal A(L)$ such that there are exactly $k$ lines above each edge (for the special case of a vertical line, we consider points to the left of it to be ``above'').
By this construction each $k$-level forms a polygonal chain.
Our data structure will store the line arrangement of a subset $S$ of lines by keeping track of all $|S|$ levels, with each level being stored in a skip list.
We will be able to support the following operations:
\begin{itemize}
    \item Answering whether a point lies on some line in $O(\log^6 n)$ time.
    \item Inserting or removing a line in $O(|S| \log^4 n + \log^6 n)$ time.
\end{itemize}

\subsection{Skip lists}

We will need a history-independent data structure which can store a set of elements and support polylogarithmic time insertion/removal operations.
For that purpose use the skip list data structure by Ambainis from the \textsc{Element Distinctness} algorithm \cite{Amb07}.
Among other applications, it was also used by \cite{ACHWZ20} for the \textsc{Closest Pair} problem, where they also gave a brief description.
Here we shortly describe only the details required in our algorithm and rely on the facts already proved in these papers.

Suppose that the skip list stores some set of elements $S \subseteq [N]$, according to some order such that comparing two elements requires constant time.
In a skip list, each element $i \in S$ is assigned an integer $\ell_i \in [0,\ldots,\ell_{\max}]$, where $\ell_{\max} = \lceil \log_2 N \rceil$.
The skip list itself then consists of $\ell_{\max}+1$ linked lists where the $k$-th list contains all $i \in S$ such that $\ell_i \geq k$.
We will call the $k$-th linked list the $k$-th layer (to not confuse them with $k$-levels).
In other words, each element $i \in S$ stores $\ell_i+1$ pointers, where the $k$-th of them points to the smallest element $j$ such that $j > i$ and $\ell_j \geq k$, or to \texttt{Null}, if there is no such $j$.
The first element of the skip list is called the \emph{head} and it only stores $\ell_{\max}$ pointers, the beginnings of each layer (it is convenient to imagine this element also storing the value $0$, which is smaller then any element of $S$).

\begin{figure}[H]
    \centering
    \includegraphics[scale=1.1]{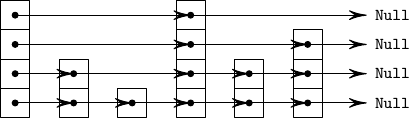}
    \caption{An example of a skip list.}
    \label{fig:skip-list}
\end{figure}

The search of an element $i \in S$ is implemented in the following way.
First, we traverse the $\ell_{\max}$-th layer to find the last element $j$ such that $j \leq i$.
If $j = i$, we are done; otherwise, traverse the layer below starting from $j$ to find the last element $j' \leq i$ there.
By repeating such iterations, we will find $i$.
Insertion of $i \notin S$ is implemented similarly: first we find the last element $j_k \in S$ such that $j < i$, for all layers $k$.
Then we update the pointers: for each layer $k \leq l_i$, we set the pointer from $i$ to be be equal to the pointer from $j_k$; then we set the pointer from $j_k$ to $i$.
If an operation requires more than $O(\log ^4 N)$ time, it is terminated.

To store the elements in memory, a specific hash table is used.
An element's entry contains the description of the element together with other data related to it (in particular, the values of the pointers).
We will not describe the details of the implementation, as it is the same as Ambainis'.
The whole data structure can sometimes malfunction (e.g., the hash table buckets can overflow or the operation of the skip list can take too long), but it is shown in \cite{Amb07} that the probability of such errors is small.
More specifically, the probability that at least one operation malfunctions among $O(N)$ operations is only $O(1/\sqrt{N})$.
Thus, as we are aiming at a sublinear algorithm, we don't need to worry about the error probability of the skip lists.
We also note that the memory requirement of Ambainis' skip list is $O(r \log^3 N)$, if at most $r$ elements need to be stored.

\subsection{Data structure}

Our data structure will operate mainly using the intersection points of the lines.
To keep a unique description of the data for history independence, we describe each intersection point in the following way.
Suppose the given lines are labeled $\ell_1$, $\ldots$, $\ell_n$.
For any two lines $\ell_i$ and $\ell_j$, let $P_{i,j}$ be their intersection point.
In an arrangement which includes both of these lines, we describe the left and right edges of $\ell_i$ connected to $P_{i,j}$ by $\texttt{left}(\ell_i,\ell_j)$ and $\texttt{right}(\ell_i,\ell_j)$.
In that case there will be a $k$ such that the $k$-level contains the edges $\texttt{left}(\ell_i,\ell_j)$ and $\texttt{right}(\ell_j,\ell_i)$, and an adjacent level contains the edges $\texttt{left}(\ell_j,\ell_i)$ and $\texttt{right}(\ell_i,\ell_j)$.
We describe $P_{i,j}$ on the $k$-level by the pair of integers $\nu_{i,j} = (i,j)$ and by $\nu_{j,i} = (j,i)$ on the adjacent level, see Figure \ref{fig:intersection}.
We call these pairs the \emph{path points} of $P_{i,j}$.
Note that we can calculate the coordinates of any path point in constant time as it description consists of the indices of the lines.

\begin{figure}[H]
    \centering
    \includegraphics[scale=0.9]{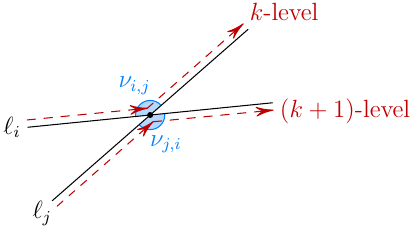}
    \caption{Path points of a line intersection.}
    \label{fig:intersection}
\end{figure}

Now we will describe the data structure, which stores an arrangement of a subset of the given lines.
It will operate with multiple skip lists each storing a set of path points of the arrangement.
To ensure the unique representation of the data, we encode the pointers of the skip lists with the values of the path points themselves.
To represent the beginning of a level from line $\ell_i$, we use a ``fictitious'' starting path point $\nu_{i,i} = (i,i)$.
The last element of the skip lists we encode with a special ``null'' path point $\nu_{\texttt{Null}}$.
We then implement the following skip lists:
\begin{itemize}
    \item The skip lists that contain the path points of the current $k$-levels in order from left to right.
    These skip lists are stored implicitly, since adding and removing lines changes the indexing of the levels and the levels themselves.
    For each path point $\nu$ stored in such a skip list, we additionally store an array $\texttt{Next}_{\nu}[0\ldots l_{\max}]$ storing the values of the next path points of its skip list for each skip list layer, up to $l_{\nu}$.
    \item $\texttt{Start}$ -- contains the heads of all level skip lists in the current arrangement.
    If the first edge of a $k$-level belongs to $\ell_i$, the head of its skip list is $\nu_{i,i}$, and we additionally store $\texttt{Next}_{\nu_{i,i}}$ to access the respective level skip list.
    The heads are ordered by the slope of the respective lines with the $x$-axis.
\end{itemize}
Further we will describe the implementation of the operations.

\paragraph{Point location.}

To detect whether a point belongs to some line, we essentially binary search through all $k$-levels and check whether the given point is strictly above, below or belongs to that level.
The binary search is essentially performed by searching in the skip list \texttt{Start}.
For a given level, we then search for its edge such that the $x$-coordinate of the point belongs to the projection of that edge on the $x$-axis.
When this edge is found, we check the relative vertical position of this point in constant time.

Essentially we have two nested searches in the skip list structure, so the complexity of this step is $O(\log^8 n)$, but we can show a better estimate.
Each time after determining that the given point is above or below a $k$-level (which takes one search in an inner skip list), we are then accessing a pointer in the outer skip list to proceed with the search.
Ambainis (\cite{Amb07}, see the proof of Lemma 6) showed that in a skip list search operation, at most $O(\log^2 n)$ pointer accesses are necessary.
Therefore, the outer search requires $O(\log^4 n)$ steps and the inner searches altogether require $O(\log^6 n)$ steps, so we improve our estimate to $O(\log^6 n)$.

\paragraph{Line insertion and removal.}

We will only describe the procedure of inserting a line in the data structure, as removing a line can be implemented by a reverse quantum circuit.
Suppose the line to be inserted is $\ell_i$; our task is to correctly update the pointers of the skip lists.
As we will see, conveniently it suffices to update only the pointers of the new edges created by the insertion of the new line.

First, we create a new skip list containing all of the path points on the inserted line.
We can do so by iterating over the existing lines in the arrangement and calculating all of the intersection points with the new line.
This can be done by iterating over the \texttt{Start} skip list, as each head keeps a reference to a single line.
We order the path points by the $x$ coordinate; in the case of path points $\texttt{left}(\ell_i,\ell_j)$ and $\texttt{right}(\ell_i,\ell_j)$ corresponding to the same $P_{i,j}$, the first comes before the second.
Additionally, we insert $\nu_{i,i}$ at the start and $\nu_{\texttt{Null}}$ at the end of the new skip list.
However, we don't insert $\nu_{i,i}$ into $\texttt{Start}$ yet.

We can see now that two consecutive path points (except the ones corresponding to the same intersection) in the new skip list define an edge in the modified arrangement.
Figure \ref{fig:line-edges} shows an example of the new edges collinear with $\ell_i$ being created.
Suppose that $P_{i,j_1}$ and $P_{i,j_2}$ are two consecutive intersection points with $\ell_i$ ($P_{i,j_1}$ is left of $P_{i,j_2}$).
Some $k$-level will pass through the edge connecting $P_{i,j_1}$ and $P_{i,j_2}$.
This level will also pass through $\texttt{left}(\ell_{j_1},\ell_i)$ and $\texttt{right}(\ell_{j_2},\ell_i)$, as all edges of a level are directed from left to right.
Therefore, the edge should connect $\nu_{j_1,i}$ with $\nu_{i,j_2}$.
There are two special cases for the first and the last edge; in the first case, the first path point is $\nu_{i,i}$ and in the second, the last path point is $\nu_{\texttt{Null}}$.

\begin{figure}[H]
    \centering
    \includegraphics[scale=0.9]{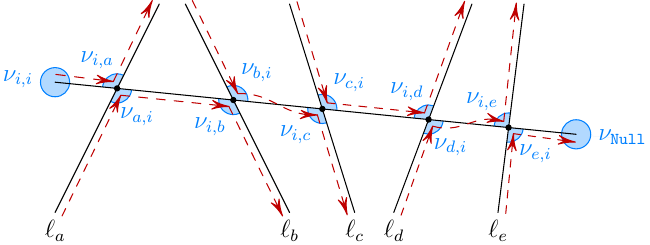}
    \caption{New edges along the inserted line.}
    \label{fig:line-edges}
\end{figure}

Next we will correct all of the level skip lists according to the updated arrangement.
Essentially, our algorithm performs a sweep line from right to left which swaps the tails of skip lists at the intersection points of $\ell_i$ with other lines.
First, we create an array containing the same set of path points as the skip list of $\ell_i$ except $\nu_{i,i}$ and $\nu_{\texttt{Null}}$, in the same order (at the end of the procedure we null the array by applying this in reverse).
We then examine the intersection points of $\ell_i$ with the other lines from right to left (we can't do this in the skip list since it's unidirectional).

Suppose we examine the intersection point $P_{i,j}$ of $\ell_i$ with $\ell_j$, see Figure \ref{fig:tail-swap}.
Then there is some edge from the old arrangement from $\nu^{(1)}$ to $\nu^{(2)}$ along $\ell_j$ which intersects $\ell_i$ at $P_{i,j}$.
The respective pair of path points is $\nu_{i,j}$ and $\nu_{j,i}$.
Then some $k$-level will pass along $\ell_i$ through $\nu_{i,j}$ and $\nu^{(2)}$, and some adjacent level (either $(k+1)$-level or $(k-1)$-level) will pass along $\ell_j$ from $\nu^{(1)}$ to $\nu_{j,i}$.
Observe that the tails of these levels (from this intersection point to the right) have been correctly updated by the sweep line.
Thus, we just need to swap the tails of these two skip lists.

\begin{figure}[H]
    \centering
    \includegraphics[scale=0.9]{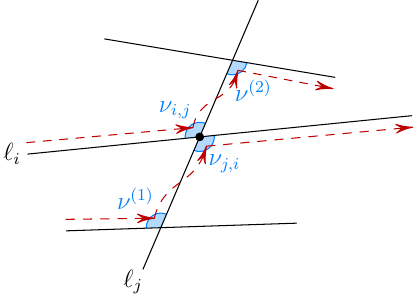}
    \caption{Updated levels at the intersection of an old line $\ell_j$ with the new line $\ell_i$.}
    \label{fig:tail-swap}
\end{figure}

To find the edge from $\nu^{(1)}$ to $\nu^{(2)}$, we use the point location operation with $P_{i,j}$.
Since we know that this point will belong to some $k$-level, we modify the point location operation so as to return the head $\nu_{h,h}$ of this $k$-level.
Note that although the level skip lists are partially updated, they still represent the old arrangement to the left of the previously examined intersection point, since the sweepline operates from right to left.
Now we have to swap the tails of the skip list with head $\nu_{i,i}$ after $\nu_{i,j}$ and the skip list with head $\nu_{h,h}$ after $\nu^{(1)}$.

Generally, suppose that we wish to swap the tails of skip lists with heads $\nu^{(a_h)}$ and $\nu^{(b_h)}$ after elements $\nu^{(a_t)}$ and $\nu^{(b_t)}$, respectively.
By searching $\nu^{(a_t)}$ in the $\nu^{(a_h)}$ skip list, we find all $l_{\max}$ path nodes $\nu^{(a_l)}$ such that $\texttt{Next}_{\nu^{(a_l)}}[l]$ points to a path node after $\nu^{(a_t)}$, for each $l \in [l_{\max}]$.
Similarly we define and find $\nu^{(b_l)}$ path nodes.
Then we simply swap the values of $\texttt{Next}_{\nu^{(a_l)}}[l]$ and $\texttt{Next}_{\nu^{(b_l)}}[l]$ for all $l \in [0,l_{\max}]$.

To conclude the procedure, we insert $\nu_{i,i}$ (together with $\texttt{Next}_{\nu_{i,i}}$) into \texttt{Start}.
As we only performed $O(r)$ skip list searching and insertion operations (swapping the tails has the same complexity as an element insertion, as it's only updating $2l_{\max}+2$ pointers) and a point location operation, the complexity of the procedure is $O(r \log^4 n + \log^6 n)$.
If $r = n^p$ for some $p > 0$, this simplifies to $O(r \log^4 n)$.

\subsection{Quantum algorithm}

We use the MNRS framework quantum walk on the Johnson graph \cite{Amb07,MNRS11}.
In this framework, we search for a marked vertex in an irreducible ergodic Markov chain on a state space $X$ defined by the transition matrix $P = (p_{x,y})_{x,y\in X}$.
For such Markov chain, there exists a unique stationary distribution $\pi : X \to \mathbb R$.
Let the subset of marked states be $M \subseteq X$.
To perform the quantum walk, the following procedures need to be implemented:
\begin{itemize}
    \item Setup operation with complexity $S$.
    This procedure prepares the initial state of the quantum walk:
    \[\ket{0}\ket{0} \mapsto \sum_{x \in X} \sqrt{\pi_x} \ket{x}\ket{0}, \]
    where $\pi_x$ is the stationary distribution of $P$.
    \item Update operation with complexity $U$.
    Required for performing a step of the quantum walk, it applies the transformations (and their inverses):
    \begin{align*}
        \ket{x}\ket{0} &\mapsto \ket{x} \sum_{y \in X} \sqrt{p_{x,y}} \ket{y}, \\
        \ket{0}\ket{y} &\mapsto \sum_{x \in X} \sqrt{p^*_{y,x}} \ket{x} \ket{y},
    \end{align*}
    where $p^*_{y,x}$ are the probabilities of the reversed Markov chain $P^*$ defined by $\pi_x p_{x,y} = \pi_y p^*_{y,x}$.
    \item Checking operation with complexity $C$.
    This procedure performs the phase flip on the marked vertices:
    \[ \ket{x}\ket{y} \mapsto \begin{cases}
        -\ket{x}\ket{y} & \text{if $x \in M$,}\\
        \ket{x}\ket{y} & \text{otherwise.}
    \end{cases}\]
\end{itemize}

We examine the Johnson graph on the state space $X$ being the set of all size $r$ subsets of $[n]$.
Two vertices $x, y \in X$ are connected in this graph if the intersection of the corresponding subsets has size $r-1$.
For the Markov chain, the transition probability is $p_{x,y} = \frac{1}{r(n-r)}$ for all edges.
Then we have the following theorem:
\begin{theorem}[Quantum walk on the Johnson graph \cite{Amb07,MNRS11}] \label{thm:quantum-walk}
Let $P$ be the random walk on the Johnson graph on size $r$ subsets of $[n]$ with intersection size $r-1$, where $r = o(n)$.
Let $M$ be either empty or the set of all size $r$ subsets that contain a fixed element.
Then there is a bounded-error quantum algorithm that determines whether $M$ is empty, with complexity
\[O\left(S + \frac{1}{\sqrt{r/n}}\left(\frac{1}{\sqrt{1/r}}\cdot U + C\right)\right) = O\left(S+\sqrt n\cdot U + \sqrt{\frac n r} \cdot C\right).\]
\end{theorem}

We can now prove our result:
\begin{theorem} \label{thm:algo2}
    There is a bounded-error quantum algorithm that solves Hopcroft's problem with $n$ lines and $m \leq n$ points in the plane in time:
    \begin{itemize}
        \item $O(n^{1/3} m^{1/2} \log^6 m)$, if $n^{2/3} \leq m$;
        \item $O(n^{2/5} m^{2/5} \log^6 m)$, if $n^{1/4} \leq m \leq n^{2/3}$;
        \item $O(n^{1/2} \log n)$, if $m \leq n^{1/4}$.
    \end{itemize}
    In particular, the complexity is $O(n^{5/6} \log^6 n)$ when $n = m$.
\end{theorem}

\begin{proof}
    By the duality of Hopcroft's problem, we can exchange $n$ and $m$; from here on, assume that $m \geq n$.\footnote{We do this, since in Theorem \ref{thm:speedup-backtracking} it is important that the number of lines is larger and here it is important that the number of points is larger, but we wish to keep the meaning of $n$ and $m$ to avoid confusion.}
    Our algorithm is a quantum walk on the Johnson graph of size $r$ subsets of the given $n$ lines.
    We will choose $r$ depending on $m$, but $r$ will always be $n^p$ for some $p > 0$.
    A set $S$ is marked if it contains a line such that there exists a point from the set of $m$ given points that belongs to this line.
    
    For the implementation, we follow the approach of \cite{ACHWZ20} for the quantum algorithm for closest points.
    For a set $S$, the state of the walk will be $\ket{S,d(S)}$, where $d(S)$ is our data structure for the line arrangement of $S$.
    We then implement the quantum walk procedures:
    \begin{itemize}
        \item For the Johnson graph, $\pi$ is the uniform distribution.
        Therefore, we first generate a uniform superposition over all subsets $S$ in $O(\log\binom{n}{r}) = O(r \log n)$ time.
        Then we create $d(S)$ by inserting all lines of $S$ into an initially empty data structure, requiring $O(r^2 \log^4 n)$ time.
        \item Suppose that $S$ and $S'$ are two size $r$ subsets with $|S \cap S'| = r-1$ so that $S' = (S \setminus \{i\}) \cup \{j\}$.
        We then represent a state $\ket{S,d(S)}\ket{S',d(S')}$ with $\ket{S,d(S)}\ket{i,j}$.
        As the Markov chain probabilities are the same for all edges, we need to implement the transition
        \[\ket{S,d(S)}\ket{0,0} \mapsto \sum_{i \in S} \sum_{j \notin S} \ket{S',d(S)'}\ket{j,i}.\]
        To do that, first we create a uniform superposition of all $i \in S$ and $j \notin S$ in $O(\log\binom{n}{r}) = O(r \log n)$ time, obtaining $\sum_{i \in S} \sum_{j \notin S} \ket{S,d(S)}\ket{i,j}$.
        Then, for fixed $i$ and $j$, we remove $\ell_i$ from $d(S)$ and insert $\ell_j$, obtaining $d(S')$; this takes $O(r \log^4 n)$ time.
        Finally, we swap the indices $i$ and $j$ in the second register in $O(\log n)$ time.
        The second transformation is implemented in the same way as for Johnson's graph, $p_{x,y} = p^*_{y,x}$.
        \item The checking operation runs Grover's search over all $m$ points and for each of them performs the point location operation.
        The complexity is $O(\sqrt{m}\log^6 n)$.
    \end{itemize}
    By Theorem \ref{thm:quantum-walk} the complexity of the algorithm is
    \[O\left(r^2 \log^4 n + \sqrt n r \log^4 n + \sqrt{\frac{n}{r}}\sqrt{m}\log^6 n\right).\]
    
    Suppose that $m^{2/3} \leq n$ and pick $r = m^{1/3}$.
    Then the second term dominates the first and we can simplify the expression to
    \[O\left(\sqrt{n}\log^4 n\left(r + \sqrt{\frac{m}{r}}\log^2 n\right)\right) = O(n^{1/2} m^{1/3} \log^6 n).\footnote{The number of logarithmic factors can be lowered slightly by optimizing the parameters, which we skip.}\]

    If we have $m^{1/4} \leq n \leq m^{2/3}$ we pick $r = (nm)^{1/5}$.
    Then we have $r \geq (n^{1+3/2})^{1/5} = \sqrt{n}$, and this time the first term in the complexity dominates the second, and the complexity is
    \[O\left(r^2 \log^4 n + \sqrt{\frac{nm}{r}}\log^6 n\right) = O(n^{2/5} m^{2/5} \log^6 n).\]

    Finally, for $n \leq m^{1/4}$, we don't have to use either the quantum walk or the history-independent data structure.
    First, we build any classical data structure for point location in a line arrangement with $O(n^2 \log n)$ build time and space and $O(\log n)$ query time (e.g.~see \cite{Ede87}, Chapter 11).
    Then we run Grover's search over all points and for each check whether it belongs to some line.
    The complexity in this case is
    \[O(n^2 \log n + \sqrt{m}(\log m + \log n)) = O(\sqrt{m} \log m).\]
    Note that we would obtain the same asymptotic complexity by using the quantum walk with $r = n$, only with more logarithmic factors.
\end{proof}

\section{Open ends}

A gap still remains between the upper and lower bounds in case $m > n^{1/4}$.
However, the lower bound is only a query lower bound; quite possibly, some overhead from data structures is necessary.
An interesting direction would be to prove some stronger lower bounds on the time complexity of Hopcroft's problem.
Perhaps one can show some fine-grained conditional lower bounds like in \cite{ACHWZ20, BLPS22a}.

\section{Acknowledgements}

This work was supported by Latvian Quantum Initiative under European Union Recovery and Resilience Facility project no.~2.3.1.1.i.0/1/22/I/CFLA/001 and QOPT (QuantERA ERA-NET Cofund).

\printbibliography

\end{document}